\documentclass[10pt]{amsart}
\usepackage{amsmath,amssymb,amsthm,graphicx,mathrsfs,url}
\usepackage[usenames,dvipsnames]{color}
\definecolor{darkred}{rgb}{0.4,0.1,0.1}
\definecolor{darkblue}{rgb}{0.1,0.1,0.4}
\definecolor{darkgrey}{rgb}{0.5,0.5,0.5}
\usepackage[colorlinks=true,linkcolor=darkred,citecolor=Blue]{hyperref}
\numberwithin{equation}{section}
\theoremstyle{plain}% default
\newtheorem{theorem}{Theorem}[section]
\newtheorem{lemma}[theorem]{Lemma}

\newtheorem{proposition}[theorem]{Proposition}

\newtheorem{dfn}[theorem]{Definition}
\theoremstyle{remark}
\newtheorem{remark}[theorem]{Remark}
\theoremstyle{plain}

\newcommand{\be}{\begin{equation}}
\newcommand{\ee}{\end{equation}}
\newcommand{\beu}{\begin{equation*}}
\newcommand{\eeu}{\end{equation*}}
\newcommand{\besu}{\begin{equation*}
\begin{aligned}}
\newcommand{\eesu}{\end{aligned}
\end{equation*}}
\newcommand{\bes}{\begin{equation}
\begin{aligned}}
\newcommand{\ees}{\end{aligned}
\end{equation}}

\newcommand\void[1]{}

\title[{ Point contacts and boundary triples}]{{\small Point contacts and boundary triples}}

\begin{document}

\author{Vladimir Lotoreichik}
\address{Institute of Computational Mathematics, 
Graz University of Technology, Steyrergasse 30,
Graz, 8010, Austria,
E-mail: lotoreichik@math.tugraz.at}

\author{Hagen Neidhardt}
\address{Weierstrass Institute
for Applied Analysis and Stochastics, Mohrenstrasse 39,
Berlin, 10117, Germany,
E-mail: Hagen.Neidhardt@wias-berlin.de}

\author{Igor Yu. Popov}

\address{Department of Higher Mathematics, St. Petersbsurg National Research University of IT, Mechanics and Optics, Kronverkskiy pr. 49, St.~Petersburg, 197101, Russian Federation,
E-mail: popov1955@gmail.com}

\begin{abstract}
We suggest an abstract approach for point contact 
problems in the framework of boundary triples. 
Using this approach  we obtain the perturbation series for a simple eigenvalue in the discrete spectrum of the model self-adjoint extension with weak point coupling.  An example
of a two-level quantum model is provided.
\end{abstract}

\maketitle
\vspace{-2mm}
\textit{\small For the proceedings of the conference QMath12, Berlin, 2013.}
%****************************************************************
\section{Introduction}
%****************************************************************

Let $H_0$ be a self-adjoint operator with an isolated simple eigenvalue $\lambda_0$. Further let $V$ be a bounded or unbounded self-adjoint operator 
such that the family of operators $H(\varkappa) := H_0 + \varkappa V$
is well-defined and self-adjoint for sufficiently small coupling constants $\varkappa \in \mathbb{R}$. If $V$ is relatively compact with respect to $H_0$, 
then there is a smooth function $\lambda(\varkappa)$ such that $\lambda(\varkappa)$ is a simple eigenvalue of $H(\varkappa)$ for each $\varkappa$ and
$\lim_{\varkappa\to 0}\lambda(\varkappa) = \lambda_0$ holds,  
cf. \cite{Kato1995}. Since the function $\lambda(\varkappa)$ 
is smooth it admits a Taylor-type expansion of the form 
\begin{equation}\label{1.0}
\lambda(\varkappa) = \lambda_0 + a\varkappa  + b\varkappa^2 + O(\varkappa^3).
\end{equation}
The problem is to compute the coefficients $a$ and $b$ of this perturbation series
in terms of the operators $H_0$ and $V$.

In a slightly modified form, similar problem 
appears for point contacts in quantum mechanics. Typically one considers 
two quantum systems which do not interact, where one of them has a simple isolated eigenvalue $\lambda_0$. If both systems are coupled by a point contact,
then the eigenvalue $\lambda_0$ can move either along the real axis or become a pole of the analytic continuation of the resolvent of the coupled system in the lower complex half-plane. In the last case one speaks of resonances. 
The eigenvalue case realizes if the second system has no spectrum around $\lambda_0$ while the resonance case appears 
if the second system has continuous spectrum around $\lambda_0$, that is, if 
$\lambda_0$ is an embedded eigenvalue for the decoupled systems.  
If the point interaction depends on a parameter $\varkappa$ such that for $\varkappa \to 0$ the coupled system converges to the decoupled one, then again a perturbation series is expected either for the eigenvalues or for the resonances.
In the following we focus on the eigenvalue case.

Perturbation series for point interactions were perhaps first studied by 
B.\,S.~Pavlov in \cite{P84,P87} for 
a model of point interactions with an inner structure, where 
the first order coefficient $a$ was computed.
A direct sum of two three-dimensional Schr\"odinger operators
coupled by a point contact was considered by P.~Exner in  \cite{E91}. In this paper he was able to compute the first and second order coefficients $a$ and $b$.
See also related work \cite{CCF09} on spin-dependent
point interactions and \cite{CCF10} 
for perturbation of eigenvalues at threshold in point contact models.
A survey on the resonance case can be found in \cite{E13}, 
see also references therein.   
Point contact models are often used in other areas of mathematical physics. 
In \cite{P93} a model of a small window in the screen is studied. 
In \cite{P12} Maxwell and Schr\"odinger operators are coupled via a point contact and  in \cite{P13} a model of a three-dimensional Helmholtz resonator is constructed via point coupling.

In the following we consider an abstract point contact model and are interested in the perturbation series for its eigenvalues. 
In particular, let $\widetilde{A}$ and $\widehat{A}$
be two densely defined closed symmetric operators 
in the Hilbert spaces $\widetilde{\mathcal H}$ and $\widehat{\mathcal H}$,
respectively, both having equal finite deficiency indices $(d,d)$.
Let us consider the direct sum $A := \widetilde{A} \oplus \widehat{A}$ 
which is also a densely defined closed symmetric operator in
$\widetilde{\mathcal H}\oplus\widehat{\mathcal H}$
with deficiency indices $(2d,2d)$. 
Further, let $\widetilde A_{[\alpha]}$ and $\widehat A_{[\beta]}$ be self-adjoint extensions of $\widetilde{A}$ and $\widehat{A}$, respectively. 
The Hamiltonian of the decoupled system is given by 
$H_0 := \widetilde A_{[\alpha]} \oplus \widehat A_{[\beta]}$,
which is a self-adjoint extension of $A$. 
As usual the Hamiltonian of the point coupled system
is given by another self-adjoint extension $H$ of $A$,
which can not be decomposed into the orthogonal sum with respect 
to the decomposition $\widetilde {\mathcal H}\oplus \widehat {\mathcal H}$. The family  $H(\varkappa)$ from above is now replaced 
by a one-parametric family of point contacts, that means,  
by a family of self-adjoint extensions $H(\varkappa)$ of $A$. 

To make the problem precise we use the framework of boundary triples. 
In this framework a subfamily $A_\Lambda$ of self-adjoint extensions  of $A$
are labeled by Hermitian matrices $\Lambda$ in ${\mathbb C}^{2d}$ via 
an abstract boundary condition involving $\Lambda$.
In particular, there is a Hermitian matrix $\Lambda_0$ 
such that $H_0 = A_{\Lambda_0}$. Let us 
assume that $\lambda_0$ is an isolated eigenvalue of 
$\widehat A_{[\beta]}$,
but a resolvent point of  $\widetilde A_{[\alpha]}$. 
Moreover, let $\Lambda(\varkappa)$ 
be a one-parametric sufficiently regular
family of Hermitian matrices in ${\mathbb C}^{2d}$,
which converges to $\Lambda_0$ as $\varkappa \to 0$. 
Setting $H(\varkappa) := A_{\Lambda(\varkappa)}$ 
one gets a family of self-adjoint extensions $H(\varkappa)$ of $A$ 
which converges in an appropriate sense to $H_0$ as $\varkappa \to 0$
and in the discrete spectra of the operators $H(\varkappa)$ 
exists a branch of the type \eqref{1.0}.
The goal is to compute the coefficients $a$ and $b$ for this branch
in terms of the Taylor coefficients of $\Lambda(\varkappa)$
and the abstract Weyl function $M(\cdot)$
which is an important ingredient of the boundary triple approach. 
In general this problem can not be reduced
to the investigation of holomorphic operator families
of the types (A) or (B) in the sense of Kato
which are thoroughly discussed  in \cite{Kato1995}. 
Only in some special cases such a reduction can be done.

We solve this problem for arbitrary $d\in{\mathbb N}$ 
for the first order coefficient $a$.
In the special case of $d = 1$
we obtain first and second order coefficients $a$ and $b$ which is beyond Pavlov~\cite{P84,P87} and covers~\cite{E91}. 
In general, it would be also possible to compute $b$ for
arbitrary finite deficiency indices, however, we have not included that
for the purpose to avoid tedious computations.

Our abstract results are illustrated with a vector two-level quantum model, which is a generalization of the analogous
scalar model considered in~\cite{E91,E13}.

%****************************************************************
\section{Boundary triples and Weyl functions}
%****************************************************************
The reader may consult with~\cite{BMN02, BGP08, DM91, DM95, MN13, S} for the theory of boundary triples and its applications. 
In this note we use this concept only in the case of finite deficiency indices.
Throughout this section the following hypothesis is employed.\\[1.0mm]
{\bf Hypothesis I.}\,{\em Let $A$ be a closed, symmetric, densely defined operator in a Hilbert space
${\mathcal H}$ with equal finite deficiency indices $(d,d)$.
}
%------------------------------------------------------------------
%------------------------------------------------------------------
\begin{dfn}
Assume that Hypothesis I holds.
The triple $\{{\mathbb C}^d,\Gamma_0,\Gamma_1\}$ with $
\Gamma_0, \Gamma_1\colon{\rm dom}\, A^*\rightarrow{\mathbb C}^d$
is a boundary triple for $A^*$ if the following conditions hold:
the mapping $\Gamma := (\Gamma_0,\Gamma_1)^\top$ is surjective onto  ${\mathbb C}^{2d}$ and the abstract Green's identity
$(A^*f,g)_{{\mathcal H}}- (f,A^*g)_{{\mathcal H}}  = (\Gamma_1f, \Gamma_0g)_{{\mathbb C}^d}- (\Gamma_0f, \Gamma_1g)_{{\mathbb C}^d}$
holds for all $f,g\in{\rm dom}\, A^*$.
\end{dfn}
%------------------------------------------------------------------
%------------------------------------------------------------------
\noindent Boundary triples is an efficient tool to parametrize self-adjoint extensions of a symmetric operator.
%------------------------------------------------------------------
%------------------------------------------------------------------
\begin{proposition}[{\cite[Proposition 1.4]{DM95},\cite[Proposition 14.7]{S}}]\label{prop:bt}
Assume that Hypothesis I holds. Let $\{{\mathbb C}^d,\Gamma_0,\Gamma_1\}$ be a boundary triple for $A^*$. Then 
for each self-adjoint extension ${\widetilde A}$ of $A$ there is a unique self-adjoint relation $\Theta$
in $\mathbb{C}^d$ such that
${\widetilde A} = A_\Theta :=
A^*\upharpoonright\{f\in{\rm dom}\, A^*
\colon \Gamma f\in\Theta\}$.
\end{proposition}
%------------------------------------------------------------------
%------------------------------------------------------------------
\begin{remark}
\label{rem:extensions}
The self-adjoint extension $A_0 := A^*\upharpoonright \ker\Gamma_0$ is distinguished. It corresponds to the self-adjoint relation
$\Theta_\infty := \left\{\begin{pmatrix}0\\h\end{pmatrix}: h \in \mathbb{C}^d\right\}$.
If $\Theta$ is the graph of a Hermitian matrix $\Lambda$ in $\mathbb{C}^d$, i.e $\Theta = \mathrm{graph}(\Lambda)$, then one easily checks that
\begin{equation}
\label{AB}
A_{\mathrm{graph}(\Lambda)} = A_\Lambda := 
A^*\upharpoonright\{f\in{\rm dom}\, A^*
\colon 
\Gamma_1f = \Lambda\Gamma_0f\}.
\end{equation}
The operator $\Lambda$ is called the boundary operator with respect to the boundary triple $\{{\mathbb C}^d,\Gamma_0,\Gamma_1\}$.
\end{remark}
%------------------------------------------------------------------
%------------------------------------------------------------------
\noindent One can associate $\gamma$-fields and Weyl functions with boundary triples.
%------------------------------------------------------------------
%------------------------------------------------------------------
\begin{dfn}
Assume that Hypothesis I holds. Let $\{{\mathbb C}^d,\Gamma_0,\Gamma_1\}$ be a boundary triple for $A^*$. The  function $\gamma \colon\rho(A_0)\rightarrow{\mathcal B}({\mathbb C}^d, {\mathcal H})$ defined as 
\[
\gamma(\lambda) := \big(\Gamma_0\upharpoonright\ker(A^*-\lambda)\big)^{-1},
\qquad \lambda\in\rho(A_0),
\]
is called the \emph{$\gamma$-field}. 
\noindent The function $M\colon\rho(A_0)\rightarrow {\mathbb C}^{d\times d}$ defined as $M(\lambda) := \Gamma_1\gamma(\lambda)$ is called the \emph{Weyl function}.
\end{dfn}
%------------------------------------------------------------------
%------------------------------------------------------------------
\begin{proposition}
\label{prop:gammaM}
Assume that Hypothesis~I holds. Let $M$ be the Weyl function associated with a boundary triple $\{{\mathbb C}^d,\Gamma_0,\Gamma_1\}$ for $A^*$. Let the self-adjoint operator $A_\Lambda$ in $\mathcal H$ 
be as in~\eqref{AB}. Then the following statements hold.
\begin{itemize}\setlength{\parskip}{1.0mm}
\item [\em (i)] The function $M(\cdot)$ is holomorphic on $\rho(A_0)$. 
\item [\em (ii)] For $\lambda\in\rho(A_0)$ the relation 
${\rm dim}\ker(A_\Lambda-\lambda) = 
{\rm \dim}\ker(\Lambda - M(\lambda))$ holds.
\item [\em (iii)] If $\lambda_0\in\rho(A_0)$ 
is a simple eigenvalue of $A_{\Lambda}$, then the function 
\[
D_{\Lambda}(\lambda) := {\rm det}(\Lambda  - M(\lambda)),
\qquad \lambda\in\rho(A_0),
\]
has a simple zero at $\lambda = \lambda_0$. 
In particular, $D'_{\Lambda}(\lambda_0) \ne 0$ holds.
\end{itemize}
\end{proposition}
%------------------------------------------------------------------
%%------------------------------------------------------------------
\begin{proof}
All the statements of this proposition are known. 
Item (i) can be found in \cite[Proposition 1.21]{BGP08}, 
see also \cite[Proposition 14.15\,(iv)]{S} and item (ii) 
is given in \cite[Theorem 1.36\,(1)]{BGP08}, 
see also \cite[Proposition 14.17\,(ii)]{S}.
For item (iii) see \cite[Corollary 4.4, Proposition 5.1\,(iii)]{MN13}.
\end{proof}
%*********************************************************************
\section{Abstract point contact and its weak coupling regime}
%*********************************************************************
In this section we present an abstract treatment of point contacts in the framework of boundary triples and obtain the perturbation series of the simple eigenvalue in the weak coupling regime. We make use of the following hypothesis.\\[1.0mm]
\noindent {\bf Hypothesis II.}\,{\em Let $\widetilde {\mathcal H}$ and $\widehat {\mathcal H}$ be separable Hilbert spaces.
Let $\widetilde A$ and $\widehat A$ be closed, densely defined, symmetric operators in  $\widetilde {\mathcal H}$ and $\widehat {\mathcal H}$, respectively, both with deficiency indices $(d,d)$. Let $\{{\mathbb C}^d,\widetilde \Gamma_0,\widetilde \Gamma_1\}$ 
and $\{{\mathbb C}^d,\widehat \Gamma_0,\widehat \Gamma_1\}$ be boundary triples for 
$\widetilde A^*$ and $\widehat A^*$, respectively. 
}
\\[1.2mm]
\noindent The next lemma appears to be useful in what follows.
%------------------------------------------------------------------
%------------------------------------------------------------------
\begin{lemma}[\!\!{\cite[Section 1.4.4]{BGP08}}] 
\label{lem:bt}
Assume that Hypothesis II holds.
Then the operator $\widetilde A\oplus \widehat A$ is closed, densely defined and symmetric in the Hilbert space $\widetilde {\mathcal H}\oplus \widehat {\mathcal H}$ with deficiency indices $(2d,2d)$ and 
$\{
{\mathbb C}^{2d},
\widetilde\Gamma_0\oplus\widehat\Gamma_0,
\widetilde\Gamma_1\oplus\widehat\Gamma_1
\}$ 
is a boundary triple for 
$(\widetilde A\oplus \widehat A)^*$.
\end{lemma}
%------------------------------------------------------------------
%------------------------------------------------------------------
Our model operator $A_{\Lambda}$ in the Hilbert space $\widetilde {\mathcal H} \oplus \widehat 
{\mathcal H}$ is defined as 
\[
\begin{split}
A_{\Lambda}
(\widetilde f\oplus \widehat f) &:= 
\widetilde A^*\widetilde f\oplus \widehat A^* \widehat f,\\
{\rm dom}\,A_{\Lambda} &:= 
\Bigg\{\widetilde f\oplus \widehat f \in 
{\rm dom}\, \widetilde A^*
\oplus
{\rm dom}\, \widehat A^*\colon  
\Lambda
\begin{pmatrix} 
\widetilde\Gamma_0\widetilde f\\ 
\widehat\Gamma_0\widehat f \end{pmatrix} 
= 
\begin{pmatrix} 
\widetilde\Gamma_1\widetilde f\\ 
\widehat\Gamma_1\widehat f 
\end{pmatrix}\Bigg\},
\end{split}
\]
with a Hermitian $2d\times 2d$ matrix of the form
\begin{equation}
\label{Lambda}
\Lambda := 
\begin{pmatrix}
\alpha I_{d} & \omega I_{d}\\
\overline\omega I_{d} & \beta I_{d}
\end{pmatrix},
\qquad
\alpha,\beta\in{\mathbb R},~\omega\in{\mathbb C}.
\end{equation}
%------------------------------------------------------------------
%------------------------------------------------------------------
\begin{proposition}
The operator $A_{\Lambda}$, defined as above, is self-adjoint in the Hilbert space $\widetilde {\mathcal H}\oplus \widehat {\mathcal H}$.
\end{proposition}
%------------------------------------------------------------------
%------------------------------------------------------------------
\begin{proof}
The statement of this proposition is a straightforward consequence of the structure of the matrix $\Lambda$, Proposition~\ref{prop:bt}, Remark~\ref{rem:extensions} and Lemma~\ref{lem:bt}.
\end{proof}
%------------------------------------------------------------------
%------------------------------------------------------------------
The next theorem contains the main results of this note: the two terms expansion of a bound state of $A_{\Lambda}$ for small coupling parameter $|\omega|$ in the case of arbitrary 
$d\in{\mathbb N}$ and the three terms analogous expansion in the special case $d = 1$.
In its formulation we use self-adjoint operators
\begin{equation}
\label{ext}
\begin{split}
\widetilde A_{[\alpha]} := 
\widetilde A^*\upharpoonright 
\ker(\widetilde\Gamma_1 - \alpha\widetilde\Gamma_0),&
\qquad 
\widetilde A_0 := 
\widetilde A^*\upharpoonright \ker\widetilde\Gamma_0,\\
\widehat A_{[\beta]} := 
\widehat A^*\upharpoonright 
\ker(\widehat\Gamma_1 - \beta\widehat\Gamma_0),&\qquad
\widehat  A_0 := 
\widehat A^*\upharpoonright \ker\widehat\Gamma_0.
\end{split}
\end{equation}

Let $L$ be a $d \times d$-matrix. In the following we use the notion of the adjugate matrix $\mathrm{adj}(L)$, cf.~\cite{Bosch2008,MN98}.
Notice that the adjugate of a matrix is quite different from the adjoint one $L^*$. 
%------------------------------------------------------------------
%------------------------------------------------------------------
\begin{theorem}
\label{thm:main}
Assume that Hypothesis II holds with some $d \in{\mathbb N}$. 
Let $\widetilde M$ and $\widehat M$ be the Weyl functions 
associated with boundary triples from that hypothesis. Let the self-adjoint operators $\widetilde A_{[\alpha]}$ and $\widehat A_{[\beta]}$ be as above. 
Assume that the real value $\lambda_0$ satisfies  $\lambda_0\in\rho(\widetilde A_0)\cap\rho(\widehat A_0)\cap\rho(\widetilde A_{[\alpha]})$ and $\lambda_0$ is a simple isolated eigenvalue of $\widehat A_{[\beta]}$.
\begin{itemize}
\item [\em (i)] 
Then for sufficiently small $|\omega|$ in the discrete spectrum of $A_{\Lambda}$ there is a branch
\begin{equation}
\label{branch}
\lambda(|\omega|^2) = \lambda_0 + a|\omega|^2 + O(|\omega|^4),\qquad |\omega|\rightarrow 0+,
\end{equation}
with 
\begin{equation}
\label{A}
a := \frac{
{\rm tr}\,
\Big(
{\rm adj}\,\big(\beta I_d -\widehat M(\lambda_0)\big)
\big(\widetilde M(\lambda_0)-\alpha I_d\big)^{-1}
\Big)}
{
{\rm tr}\,\Big( 
{\rm adj}\,\big(\beta I_d -\widehat M(\lambda_0)\big)
\widehat{M}'(\lambda_0)\Big)},
\end{equation}
where ${\rm adj}\,(\beta I_d -\widehat M(\lambda_0))$ is the 
adjugate matrix.
%--------------------------------------------------------------
\item [\em (ii)]
Suppose that $d =1$. Then the expansion \eqref{branch} can be
extended as
\begin{equation}
\label{branch2}
\lambda(|\omega|^2) = \lambda_0 + a|\omega|^2 + b|\omega|^4+  O(|\omega|^6),\qquad |\omega|\rightarrow 0+,
\end{equation}
with
\begin{equation}
\label{A1}
a := 
\frac{1}{(\widetilde{M}(\lambda_0) - \alpha)\widehat{M}'(\lambda_0)}
\end{equation}
and
\begin{equation}
\label{B}
b := 
\frac{1}{((\widetilde{M}(\lambda_0) - \alpha)\widehat{M}'(\lambda_0))^2}
\left(
\frac{\widetilde{M}'(\lambda_0)}
{\alpha-\widetilde{M}(\lambda_0)} 
-\frac{1}{2}
\frac{\widehat{M}''(\lambda_0)}{\widehat{M}'(\lambda_0)}
\right).
\end{equation}
\end{itemize}
\end{theorem}
%------------------------------------------------------------------
%------------------------------------------------------------------
\begin{proof}
(i)
The proof of this item is carried out in three steps.\\
\noindent {\em Step I}.
For sufficiently small $\varepsilon > 0$ the interval $I := (\lambda_0 -\varepsilon,\lambda_0 +\varepsilon)$ is contained 
in the set $\rho(\widetilde A_0)\cap\rho(\widehat A_0)$. By Proposition~\ref{prop:gammaM}\,(i) the following matrix-valued function
\[
T(\lambda) := (\alpha I_d -\widetilde M(\lambda))(\beta I_d -\widehat M(\lambda))
\]
is well-defined and $C^\infty$-smooth on $I$. 
Next we introduce the scalar-valued function
\begin{equation}
\label{F}
F\colon I\times{\mathbb R}\rightarrow {\mathbb R},\quad F(\lambda,x) := {\rm det}\,\big(T(\lambda)- x I_d\big),
\end{equation}
which is $C^\infty$-smooth on $I\times{\mathbb R}$. 
\\
%------------------------------------------------------------------------
\noindent {\em Step II}.
The following two functions
\[
\widetilde{D}_\alpha(\lambda) := 
{\rm det}\,\big(\alpha I_d - \widetilde{M}(\lambda)\big)
\quad\text{and}\quad
\widehat{D}_\beta(\lambda) := 
{\rm det}\,\big(\beta I_d - \widehat{M}(\lambda)\big)
\]
are well-defined and $C^\infty$-smooth on $I$.
Jacobi's formula~\cite{G72, MN98} and the identity 
${\rm adj}\,(L_1L_2)= {\rm adj}\,(L_2)\,{\rm adj}\,(L_1)$ imply
%%--
\[
F_{x}(\lambda_0,0) = -{\rm tr}\,\Big(
{\rm adj}\,\big(\beta I_d -\widehat{M}(\lambda_0)\big)
{\rm adj}\,\big(\alpha I_d -\widetilde{M}(\lambda_0)\big)
\Big).
\]
%&--
In view of 
$\lambda_0\in\rho(\widetilde{A}_{[\alpha]})$ and of Proposition~\ref{prop:gammaM}\,(ii)
the matrix $\alpha I_d - \widetilde M(\lambda_0)$ is invertible.
For any invertible matrix $L$ the identity
${\rm adj}\,(L) = {\rm det}\,(L)\,L^{-1}$ holds.
Hence, we arrive at
\begin{equation}
\label{deriv1}
F_{x}(\lambda_0,0) = 
\widetilde{D}_{\alpha}(\lambda_0)
{\rm tr}\,
\Big(
{\rm adj}\,\big(\beta I_d -\widehat{M}(\lambda_0)\big)
\big(\widetilde{M}(\lambda_0)-\alpha I_d\big)^{-1}
\Big).
\end{equation}
Note that 
$F(\lambda,0) = 
\widetilde{D}_\alpha(\lambda)
\widehat{D}_\beta(\lambda)$, 
where  the identity 
${\rm det}\,(L_1L_2) = 
{\rm det}\,(L_1)\,{\rm det}\,(L_2)$ 
is used. In view of $\lambda_0\in\sigma_{\rm d}(\widehat{A}_{[\beta]})$ and of Proposition~\ref{prop:gammaM}\,(ii) we get 
$\widehat{D}_\beta(\lambda_0) = 0$, which
implies $F(\lambda_0,0) = 0$.
Next we compute $F_\lambda$ at the point $(\lambda_0,0)$
\begin{equation}
\label{deriv2}
\begin{split}
F_{\lambda}(\lambda_0,0) &= 
\Big(\frac{d}{d\lambda}F(\lambda,0)\Big)
\Big|_{\lambda = \lambda_0} \\
&=
\widetilde{D}_\alpha'(\lambda_0)
\widehat{D}_\beta(\lambda_0)
+
\widetilde{D}_\alpha(\lambda_0)
\widehat{D}_\beta'(\lambda_0) =
\widetilde{D}_\alpha(\lambda_0)
\widehat{D}_\beta'(\lambda_0).
\end{split}
\end{equation}
Since the eigenvalue $\lambda_0$ is simple in the spectrum of $\widehat{A}_{[\beta]}$, by Proposition~\ref{prop:gammaM}\,(iii)  $\widehat{D}'_\beta(\lambda_0) \ne 0$ holds. 
Similarly $\widetilde{D}_\alpha(\lambda_0) \ne 0$ because of $\lambda_0\in\rho(\widetilde{A}_{[\alpha]})$. Hence we obtain 
that $F_{\lambda}(\lambda_0,0) \ne 0$.  Recall that $F(\lambda_0,0) = 0$ and  that $F$ is $C^\infty$-smooth. Therefore, by the classical implicit function theorem~\cite[Theorem 3.3.1]{KP} there exists the $C^\infty$-smooth function $\lambda(\cdot)$ defined on a sufficiently small neighborhood of the origin such that $\lambda(0) = \lambda_0$ and that $F(\lambda(x),x) = 0$ holds pointwise. The derivative of $\lambda(\cdot)$ is given as usual by 
\begin{equation}
\label{lambda'}
\lambda'(x) = 
-\frac{F_{x}(\lambda(x),x)}
{F_{\lambda}(\lambda(x),x)}.
\end{equation}
Again using Jacobi's formula we get
\begin{equation}
\label{D'}
\widehat{D}'_\beta(\lambda_0) = 
-{\rm tr}\,
\Big(
{\rm adj}\,
\big(\beta I_d -\widehat{M}(\lambda_0)\big)
\widehat{M}'(\lambda_0) 
\Big).
\end{equation}
Substituting \eqref{deriv1}, \eqref{deriv2} and \eqref{D'} into  \eqref{lambda'} we arrive at $\lambda'(0) = a$ with $a$ given by \eqref{A}.
Hence we obtain that
\begin{equation}
\label{lambdaexp}
\lambda(x) = 
\lambda_0 
+ 
ax 
+
O(x^2),\qquad x\rightarrow 0.
\end{equation}
%------------------------------------------------------------------
\noindent {\em Step III.}
By Proposition~\ref{prop:gammaM}\,(ii) a point $\lambda\in\rho(\widetilde A_0)\cap\rho(\widehat A_0)$ satisfying
\[
{\rm det}\,
\begin{pmatrix}
\alpha I_d -  \widetilde M(\lambda) & 
\omega I_d\\
\overline\omega I_d & 
\beta I_d -\widehat M(\lambda)
\end{pmatrix} = 0
\]
is in the discrete spectrum of $A_{\Lambda}$. 
By \cite[Theorem 3]{Silvester2000} one gets that 
\[
{\rm det}\,
\begin{pmatrix}
\alpha I_d -  \widetilde{M}(\lambda) & 
\omega I_d\\
\overline\omega I_d & 
\beta I_d -\widehat M(\lambda)
\end{pmatrix} = \mathrm{det}((\alpha I_d - \widetilde{M}(\lambda))(\beta I_d - \widehat{M}(\lambda)) - |\omega|^2 I_d).
\]
That is $\lambda\in\rho(\widetilde A_0)\cap\rho(\widehat A_0)$ satisfying $F(\lambda,|\omega|^2) = 0$ with $F$ as in \eqref{F} belongs to the discrete spectrum of $A_{\Lambda}$. 
Hence for sufficiently small $|\omega|^2$
we have $\lambda(|\omega|^2)\in\sigma_{\rm d}(A_{\Lambda})$ with $\lambda(\cdot)$ defined by Step II. Finally, the expansion \eqref{lambdaexp} implies \eqref{branch} in the formulation of the theorem. 
%-------------------------------------------------------------
%-------------------------------------------------------------

\noindent (ii)
The proof of this item goes along the lines of the proof of (i) and we indicate only the  differences.
Let $F$ be defined as in \eqref{F}. In this special case ($d = 1$)
we have
\[
F(\lambda,x) = \big(\alpha  -\widetilde M(\lambda)\big)
			   \big(\beta -\widehat M(\lambda)\big)  - x. 
\] 
The 1st and 2nd order partial derivatives of $F$ 
are computed below 
\begin{equation}
\label{deriv1(1)}
\begin{split}
F_{x}(\lambda,x) &= -1,\quad 
F_{\lambda x}(\lambda,x) = 0,\quad 
F_{xx}(\lambda,x) =0,\\
F_{\lambda}(\lambda,x) &= 
 - \widetilde{M}'(\lambda)
(\beta -\widehat{M}(\lambda)) - 
(\alpha -\widetilde{M}(\lambda))
\widehat{M}'(\lambda),\\
F_{\lambda\lambda}(\lambda,x) &=
-\widetilde{M}''(\lambda)
(\beta - \widehat{M}(\lambda)) 
+ 
2\widetilde{M}'(\lambda)\widehat{M}'(\lambda)
- 
(\alpha - \widetilde{M}(\lambda))
\widehat{M}''(\lambda).
\end{split}
\end{equation}
In particular, we have at the point $(\lambda_0,0)$
\begin{equation}
\label{deriv2(1)}
\begin{split}
F_{\lambda}(\lambda_0,0) &= 
(\widetilde{M}(\lambda_0)-\alpha)\widehat{M}'(\lambda_0),\\
F_{\lambda\lambda}(\lambda_0,0) &= 2\widetilde{M}'(\lambda_0)\widehat{M}'(\lambda_0) 
+ 
(\widetilde{M}(\lambda_0)-\alpha)
\widehat{M}''(\lambda_0),
\end{split}
\end{equation}
where we used that $\beta - \widehat M(\lambda_0) = 0$, which is true in view of $\lambda_0\in\sigma_{\rm d}(\widehat A_{[\beta]})$. Similarly as on Step II in the proof of (i) we get that $F(\lambda_0,0) = 0$ and $F_\lambda(\lambda_0,0)\ne 0$. Hence, there exists the $C^\infty$-smooth function $\lambda(\cdot)$ defined on a sufficiently small neighborhood of the origin such that $\lambda(0) =\lambda_0$, that $F(\lambda(x),x) = 0$ holds pointwise and that  $\lambda'(x)$ is as in \eqref{lambda'}.
Substituting the identity $F_x(\lambda(x),x) = -1$
into \eqref{lambda'} we obtain that
\begin{equation}
\label{lambda'(1)}
\lambda'(x) = \frac{1}{F_\lambda(\lambda(x),x)},
\end{equation}
and further substituting \eqref{deriv2(1)} into the above formula we get $\lambda'(0) = a$ with $a$ as in \eqref{A1}. 
Taking the derivative in \eqref{lambda'(1)} we get
\[
\lambda''(x)  =  
-\frac{F_{\lambda\lambda}(\lambda(x),x)\lambda'(x)
+ F_{\lambda x}(\lambda(x),x)}
{(F_{\lambda}(\lambda(x),x))^2}.
\]
Plugging \eqref{deriv1(1)} and \eqref{deriv2(1)} into the above formulae we obtain $\lambda''(0) = 2b$ with $b$ as in \eqref{B}. Hence we arrive at the expansion
\begin{equation*}
\label{lambdaexp2}
\lambda(x) = 
\lambda_0 
+ 
ax 
+ 
bx^2 
+ 
O(x^3),
\qquad x\rightarrow 0,
\end{equation*}
which implies \eqref{branch2} similarly as on Step III in the proof of (i) the expansion \eqref{lambdaexp} implied the formula \eqref{branch}.
\end{proof}
%-------------------------------------------------------------
%-------------------------------------------------------------
\begin{remark}
The roles of the operators $\widetilde A_{[\alpha]}$ and $\widehat A_{[\beta]}$ in the above theorem can be interchanged.
\end{remark}
\begin{remark}
Note that ${\rm adj}\,(0) = 1$ and in the special case $d = 1$ the formula \eqref{A} reduces to \eqref{A1}.
\end{remark}

\section{An example}
\label{sec:example}
Let the operator $A$ in the Hilbert space
$L^2({\mathbb R}^3)\otimes{\mathbb C}^d$ with 
$d \in{\mathbb N}$ be defined as
\begin{equation}
\label{Adef}
A f := (-\Delta + Q) f,\qquad
{\rm dom}\, A := \big\{f\in H^2({\mathbb R}^3)\otimes
{\mathbb C}^d\colon 
f(0,0,0) = {\bf 0}\big\},
\end{equation}
where $Q = Q^*$ is a $d\times d$ matrix.
The operator $A$ is closed, symmetric, densely defined with deficiency indices $(d,d)$,
cf.~\cite{AGHH05,BMN08,GMZ12} and \cite[Section 3]{BNP13}.
The adjoint of the above symmetric
operator can be characterized according to~\cite{BMN08,GMZ12} as
\[
\begin{split}
{\rm dom}\,A^*&= 
\Big \{f = f_0 + \vec a\tfrac{e^{-|x|}}{|x|} + \vec be^{-|x|}
\colon f_0 \in {\rm dom}\, A, \vec a,\vec b\in{\mathbb C}^{d}\Big\},\\
A^*f & = -\Delta f_0 - \vec a\tfrac{e^{-|x|}}{|x|} -
\vec b\Big(e^{-|x|} - \tfrac{2e^{-|x|}}{|x|}\Big)  + Qf.
\end{split}
\]
The triple $\{{\mathbb C}^d,\Upsilon_0,\Upsilon_1\}$
with  $\Upsilon_0,\Upsilon_1\colon {\rm dom}\,A^*\rightarrow
{\mathbb C}^d$, where 
\[
\Upsilon_0 f := \sqrt{4\pi}\vec a\quad\text{and}\quad
\Upsilon_1f := \sqrt{4\pi}\lim_{|x|\rightarrow 0}\big(f(x) - 
\tfrac{\vec a}{|x|}\big)
\]
is a boundary triple for $A^*$. By 
\cite[Proposition 4.1\, (iii)]{GMZ12} and
\cite[Proposition 3.3]{BNP13} 
the Weyl function  associated with the boundary triple
$\{{\mathbb C}^d,\Upsilon_0,\Upsilon_1\}$ is given by
$
M(\lambda) = i\sqrt{\lambda - Q}$, $\lambda \in
\rho(A_0)$.  Let us assume that $Q$ has only simple eigenvalues $q_1 <
\ldots < q_k < \ldots < q_d$. It turns out that $\sigma(A_0) =
[q_1,\infty)$.  Consider the self-adjoint extension 
$A_{[\alpha]} = A^*\upharpoonright\ker(\Upsilon_1 - \alpha\Upsilon_0)$
of $A$  with $\alpha < 0$. The extension $A_{[\alpha]}$ has below
the threshold $q_1$ at most $d$ eigenvalues. Moreover, $\lambda \in
(-\infty,q_1)$ is an eigenvalue if and only if $\lambda + \alpha^2 \in \sigma(Q)$. 
Hence, $\lambda_k := q_k - \alpha^2$ is an eigenvalue of $A_{[\alpha]}$
below the threshold $q_1$ if and only if $q_k - \alpha^2 < q_1$. 
In particular, we have $A_{[\alpha]} \ge
q_1- \alpha^2$.

Suppose that symmetric operators $\widetilde A$ and 
$\widehat A$ are as in \eqref{Adef} with $Q = \widetilde Q$
and $Q = \widehat Q$, respectively.
Let us assume that $\{\widetilde q_1,\widetilde
q_2,\ldots,\widetilde q_{d}\}$ and $\{\widehat q_1,\widehat
q_2,\ldots,\widehat q_{ d}\}$ are the simple eigenvalues of $\widetilde Q$ and
$\widehat Q$, respectively, ordered increasingly. 
We denote the instances of the triple $\{{\mathbb
  C}^d,\Upsilon_0,\Upsilon_1\}$ for 
$\widetilde A^*$ by $\{{\mathbb C}^d,\widetilde\Gamma_0,\widetilde\Gamma_1\}$ and for $\widehat A^*$ by 
$\{{\mathbb C}^d,\widehat\Gamma_0,\widehat\Gamma_1\}$.
The corresponding Weyl functions are clearly given by 
\[
\widetilde M(\lambda) = i\sqrt{\lambda- \widetilde Q}
\quad\text{and}\quad \widehat M(\lambda) = 
i\sqrt{\lambda- \widehat Q}.
\]
Note that the triple $\{{\mathbb C}^{2d},\Gamma_0,\Gamma_1\}$
with $\Gamma_i = \widetilde \Gamma_i \oplus \widehat \Gamma_i$,
$i=0,1$, is a boundary triple for $(\widetilde A\oplus \widehat A)^*$.  
The self-adjoint extensions $\widetilde A_{[\alpha]}$, $\widehat A_{[\beta]}$ ($\alpha,\beta < 0$) 
and $\widetilde A_0$, $\widehat A_0$ are defined as in \eqref{ext}.
The self-adjoint extension
$\widetilde A_{[\alpha]}$ of $\widetilde A$ satisfies 
$\widetilde A_{[\alpha]} \ge \widetilde q_1 - \alpha^2$.
The self-adjoint extension
${\widehat A}_{[\beta]}$ of $\widehat A$ has the spectrum 
$\sigma(\widehat A_{[\beta]})= \big(\cup_{k=1}^l
\{\widehat  q_k - \beta^2\}\big) \cup \big[\widehat q_1,+\infty\big)$,
where the eigenvalues are simple and $l$ is the greatest
integer $l \in \{1,\ldots,d\}$ satisfying 
$\widehat q_l - \beta^2 < \widehat q_1$.  
If the condition
\[
\widehat\lambda_k :=  \widehat q_{k} - \beta^2  < 
\widetilde q_1 - \alpha^2 
%\quad\text{and}\quad
%\widehat \lambda_k  < \widehat q_1
\]
holds for some $k\le l$, then $\widehat \lambda_k$ is a simple isolated eigenvalue
of $\widehat A_{[\beta]}$ and simultaneously  a
resolvent point of $\widetilde A_{[\alpha]}$. 
Consider the self-adjoint extension 
$A_{\Lambda} := (\widetilde A\oplus\widehat A)^*\upharpoonright
\ker(\Gamma_1 - \Lambda\Gamma_0)$
of $\widetilde A\oplus\widehat A$ with $\Lambda$ as in \eqref{Lambda}. 
Simple computations give us
\[
 \widehat M'(\widehat \lambda_k) 
= \frac{i}{2}\big(\widehat \lambda_k  - \widehat Q\big)^{-1/2} = 
\frac{1}{2}\big(\widehat Q - \widehat \lambda_k\big)^{-1/2}. \\
\]
With the above formula in hands we get using
Theorem~\ref{thm:main}\,(i) 
that in the discrete spectrum of $A_{\Lambda}$ exists a branch with
the expansion
\[
\widehat \lambda_k(|\omega|^2) = \widehat \lambda_k  
- 2\frac{
{\rm tr}\,\Big(
{\rm adj}\,\big(\sqrt{\widehat Q - \widehat \lambda_k} + \beta\big)
\big(\sqrt{\widetilde Q - \widehat \lambda_k} + \alpha\big)^{-1}\Big)}{
{\rm tr}\,
\Big({\rm adj}\,\big(\sqrt{\widehat Q - \widehat \lambda_k} + \beta\big)
\big(\widehat Q - \widehat \lambda_k\big)^{-1/2} \Big)}|\omega|^2
 + O(|\omega|^4),\quad \omega\rightarrow 0.\]
Notice that $\widehat Q - \widehat \lambda_1
\ge 0$ and $\widetilde Q - \widehat \lambda_1 \ge 0$. Since
$\sqrt{\widetilde Q - \widehat \lambda_1} + \alpha > 0$  and furthermore
${\rm adj}\,\big(\sqrt{\widehat Q - \widehat \lambda_1} + \beta\big) \ge 0$
we get that $\widehat \lambda_1(|\omega|^2) < \widehat \lambda_1$ for sufficiently small $\omega$. 

Next we consider the special case $d = 1$. Setting 
$\widetilde q := \widetilde q_1$ and $\widehat q :=  \widehat q_1$
we get $\sigma(\widetilde A_{[\alpha]})= 
\{\widetilde  q -\alpha^2\} \cup \big[\widetilde q,+\infty\big)$
and $\sigma(\widehat A_{[\beta]})= 
\{\widehat  q -\beta^2\} \cup \big[\widehat q,+\infty\big)$.
Let $\lambda_0 := \widehat  q -\beta^2$. If 
$\widehat  q -\beta^2 < \widetilde q$ and 
$\widetilde q - \alpha^2 \not=  \widehat q -\beta^2$, then $\lambda_0 \in \rho(\widetilde A_{[\alpha]})$. 

Let $E := \widehat q - \widetilde q$. Notice that $\beta^2 - E > 0$. 
Simple computations give us
\[
\begin{split}
\widetilde M(\lambda_0)  = -\sqrt{\beta^2-E},&\qquad
\widetilde M'(\lambda_0) = 
\frac{1}{2\sqrt{\beta^2 - E}},\\
\widehat M'(\lambda_0) =
\frac{1}{2|\beta|},&\qquad
\widehat M''(\lambda_0) =  \frac{1}{4|\beta|^3}.
\end{split}
\]
Hence, according to Theorem~\ref{thm:main}\,(ii)
in the discrete spectrum of $A_\Lambda$ exists a branch 
with the expansion
\begin{displaymath}
\begin{split}
\lambda(|\omega|^2) = \lambda_0 +&
\frac{2\beta}
{\sqrt{\beta^2 - E} +\alpha}|\omega|^2 +\\
&\frac{1}
{\Big(\sqrt{\beta^2 - E} + \alpha\Bigr)^3}
\frac{\beta^2 + E - \alpha \sqrt{\beta^2 - E}}{\sqrt{\beta^2 - E}}
|\omega|^4+
O(|\omega|^6)
\end{split}
\end{displaymath}
as $\omega\rightarrow 0$.
%If we set $a = \beta$ and $b = \alpha$, then 
The latter result is consistent with \cite[Theorem 3.1]{E91}, 
see also~\cite{E13}.

\section*{Acknowledgment}
The work of VL was supported by the Austrian Science Fund (FWF),
project P 25162-N26. VL thanks Weierstrass Institute for Applied Analysis and Stochastics for hospitality.
The work of IYP was supported by the Government of
Russian Federation (grant 074-U01)  and by State contract of the
Russian Ministry of Education and Science. IYP thanks TU Graz for
hospitality. Jussi Behrndt and Jonathan Rohleder are acknowledged for
interesting and fruitful discussions. The authors are very grateful to
the anonymous referee for useful suggestions.

\end{document}